\documentclass[aps,pra,10pt,a4paper,reprint,nofootinbib,superscriptaddress]{revtex4-2}
\usepackage{xcolor}
\usepackage[utf8]{inputenc}
\usepackage{amsfonts}
\usepackage{amssymb}
\usepackage{amsmath}
\usepackage{amsthm}
\usepackage{graphics}
\usepackage{graphicx}
\usepackage{braket}
\usepackage[normalem]{ulem}
\usepackage[colorlinks=true,linkcolor=blue,urlcolor=blue,citecolor=blue]{hyperref}
\usepackage{times,txfonts}
\newtheorem{theorem}{Theorem}
\newtheorem{corollary}{Corollary}
\newtheorem{definition}{Definition}
\newtheorem{proposition}{Proposition}
\newtheorem{remark}{Remark}

\theoremstyle{definition}
\newtheorem{example}{Example}

\begin{document}

\title{Unconditionally superposition-robust entangled state in all multiparty quantum systems}

\author{Swati Choudhary}
\affiliation{Harish-Chandra Research Institute,  A CI of Homi Bhabha National Institute, Chhatnag Road, Jhunsi, Prayagraj  211 019, India}
\affiliation{Center for Quantum Science and Technology (CQST) and  
Center for Computational Natural Sciences and Bioinformatics (CCNSB),
International Institute of Information Technology Hyderabad, Prof. CR Rao Road, Gachibowli, Hyderabad 500 032, Telangana, India}

\author{Ujjwal Sen}
\affiliation{Harish-Chandra Research Institute,  A CI of Homi Bhabha National Institute, Chhatnag Road, Jhunsi, Prayagraj  211 019, India}

\author{Saronath Halder}
\affiliation{Harish-Chandra Research Institute,  A CI of Homi Bhabha National Institute, Chhatnag Road, Jhunsi, Prayagraj  211 019, India}
\affiliation{Center for Theoretical Physics, Polish Academy of Sciences, Aleja Lotnik\'{o}w 32/46, 02-668 Warsaw, Poland}

\begin{abstract}
We investigate the inseparability of states generated by superposition of a multipartite pure entangled state with a product state. In particular, we identify specific multipartite entangled states that will always produce inseparability after superposition with an arbitrary completely product state. Thus, these entangled states are unconditionally superposition-robust, and we refer to this phenomenon as ``unconditional inseparability of superposition" in multipartite quantum systems. In this way, we complete the picture of unconditional inseparability of superposition which was introduced earlier for bipartite systems. We also characterize the superposition of a multipartite pure entangled state with a pure bi-separable state. However, the present analysis allows us to obtain a more general result, viz., there exists at least one pure entangled state in any multipartite Hilbert space such that its superposition with an arbitrary completely product state always yields an entangled state. On the other hand, for any bipartite pure entangled state, if at least one of the subsystems is a qubit, then it is always possible to find a suitable product state such that the superposition of the entangled state and the product state produces a product state. In this way, we find a feature of multipartite quantum systems which is in sharp contrast with that of bipartite ones. We then show how unconditional inseparability of superposition can be useful in exhibiting an indistinguishability property within the local unambiguous state discrimination problem.
\end{abstract}
\maketitle

\section{Introduction}\label{sec:Introduction}
The superposition principle is a fundamental concept in quantum theory \cite{aberg2006,TMM2014, AGM2017}. It serves as a prerequisite for the manifestation of another fundamental concept, viz., entanglement \cite{RPMK2009, GUHNE2009, STMAAU17}, a distinct quantum resource. This resource offers significant advantages in various information processing tasks such as quantum teleportation \cite{BBCJPW1993, MDMV1999, MJT2004}, dense coding \cite{CS1994, AA1994, KHPN1996, DGMCAU2004, RAAU2013, RAU2013}, quantum key distribution \cite{A1991, CG2014, S2022} etc. Clearly, given such advantages in information processing tasks, characterization of entanglement is quite necessary to understand the maximum utility of such a resource. Furthermore, a study which involves both concepts -- the superposition principle and entanglement -- is potentially significant. See Refs.~\cite{AUHMG2015, EC2016,HZZSV2017,AU2021,ASU2023} and \cite{ARARCPKDA2007,ZGZ2010,D2013,NFM2016,AESMAM2016,XTF2017,KH2018,LAJSRZCXC2018,HML2018,YTZXS2019,TSM2020,LMJY2021} in this regard. 

In the presence of noise, an efficient way of characterizing  entanglement is described as the following. We usually consider convex combination of an entangled state and a separable state. Then, we check up to what extent the entangled state can tolerate noise, i.e., what are the instances for which the state after convex combination remains entangled. This is referred to as the study of robustness of entanglement \cite{Vidal1999}. Interestingly, if we consider a convex combination of an arbitrary pair of a pure bipartite entangled state and a product state then, the resultant state is always entangled \cite{HORODECKI2003589}; provided that the probability corresponding to the entangled state is nonzero. This gives rise to an instance of \textit{unconditional robustness of entanglement} -- any pure entangled state is unconditionally robust in the direction of any product state. However, one can easily extend this to the multipartite scenario. For this, we can provide the following reason. Any pure entangled state in a multipartite system is entangled in at least one bipartition. Then, we focus on that bipartition and take a convex combination of the entangled state with an arbitrary completely product state (such a state is product across every bipartition), provided the probability corresponding to the entangled state is nonzero. In this way, the problem reduces to the bipartite case. Consequently, this convex combination results in an entangled state. 

An almost analogous concept to ``unconditional robustness of entanglement'' is ``unconditional inseparability of superposition'', introduced in \cite{sep+ent} considering bipartite entangled states. In the latter case, the relevant question is the following: Under what conditions, the superposition of a pure entangled state and a product state always produces an entangled state? Here, the properties of the entangled state might be given but the form of the product state is not known -- which means that we are checking if the entangled state with the given properties can be unconditionally superposition-robust in the direction of any product state. The nontrivial nature of this exploration is underlined by the fact that  unlike convex combinations, linear combinations (i.e., superpositions) of any pure entangled state and an arbitrary product state cannot always produce an entangled state, provided that the combination is a nontrivial one. We mention that in a ``nontrivial'' linear combination (superposition), the coefficients corresponding to the states are nonzero. There are several articles where superposition of quantum states is considered and then, the researchers have provided bounds on different entanglement measures both in bipartite and in multipartite systems \cite{LPS2006, UH2007, CXH2007, JN2007, G2007, DMA2007, WNZ2007, GA2008, AJA2008, GHIU2009, KCH2010, ZZF2010, S2011, PS2011, AAD2011, ZZS2014, CTD2016, QCFXFQ2018, SLZ2019, TDC2019}. However, instances where these important quantitative studies fail, it may still be possible, via independent methods, to know whether the output of the superposition is entangled or not. Some explicit examples in this direction were constructed for bipartite quantum systems in \cite{sep+ent}, where it was also shown that such qualitative studies, in the absence of a quantitative one, can have important applications in information processing tasks. In this work, we complete the picture of exploring unconditional inseparability of superposition considering multipartite quantum systems. We consider a superposition of a pure multipartite entangled state and a product state (in this context, see Fig.~\ref{fig:quantum_superposition}). Then, we explore the instances when unconditional inseparability of superposition can be observed. More precisely, we ask about the properties of the entangled state which can guarantee that the state after superposition is entangled, no matter what is the form of the product state. Clearly, this is a study on the interrelation between two fundamental concepts, viz.,  the superposition principle and entanglement. Furthermore, we show that this study has a connection with quantum  state discrimination problems under local quantum operations and classical communication (LOCC) in an unambiguous state discrimination setting. 

\begin{figure}[t]
\centering  
\vspace{-2cm}
\includegraphics[scale=0.29]{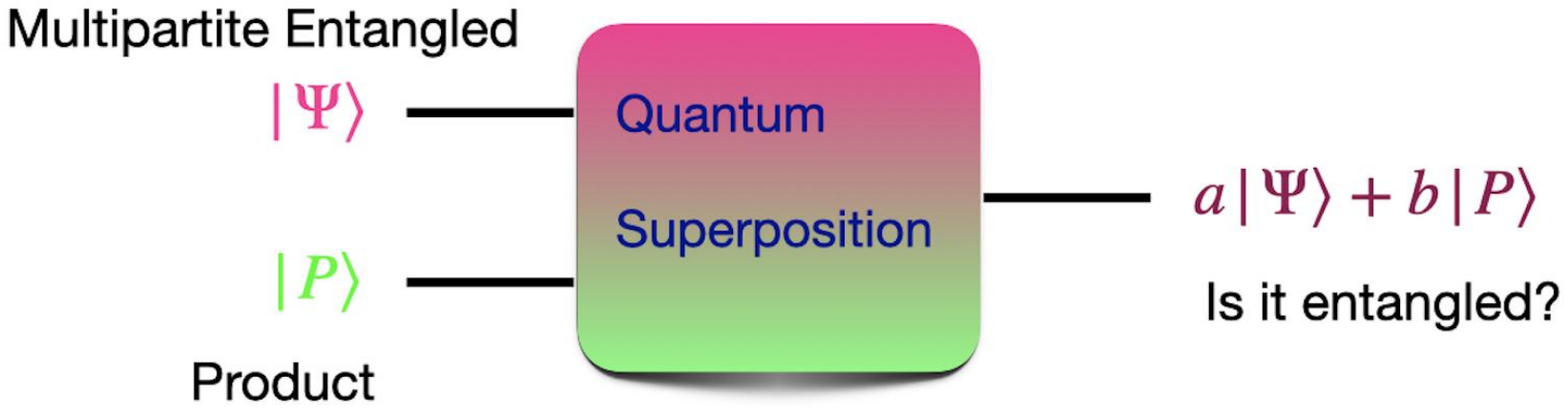}
\vspace{-2cm}
\caption{A schematic diagram of quantum superposition depicting the instance that is considered in this paper.}\label{fig:quantum_superposition}  
\end{figure}

In a state discrimination problem under LOCC, we seek the knowledge of the state in which the quantum system is prepared by doing LOCC, provided that the state is chosen from a known set. Furthermore, when we consider the unambiguous setting, we try to seek the knowledge of the state of the system with nonzero probability without committing any error. See Refs.~\cite{CHEFLES2004, Duan07, BJ2009, Cohen15, Ha21, SU2022} in this context. In our work, we show that the concept of ``unconditional inseparability of superposition'' can be useful to exhibit indistinguishability of orthogonal quantum states under LOCC in an unambiguous setting. In particular, we present a class of state discrimination problems under LOCC where unconditional inseparability of superposition can provide sufficient conditions for exhibition of  indistinguishability of states of the given set, under LOCC in an unambiguous setting.

In brief, the present work is important for the following reasons. (i) We complete the picture of unconditional inseparability of superposition, considering multipartite generalization of what was obtained in \cite{sep+ent}. (ii) In particular, such a generalization to multiparty systems is nontrivial, because in multipartite systems, there are different types of product states, for example, completely product states, bi-separable states, etc. Specifically, we consider two different cases: (a) superposition of a multipartite pure entangled state and a completely product state, and (b) superposition of a multipartite pure entangled state (preferably, the state is entangled across every bipartition) and a pure bi-separable state. From now on, in this paper when we refer to a  bi-separable state, it is  a pure bi-separable state. (iii) We find that unconditional inseparability of superposition can be found in all multipartite Hilbert spaces. Actually, for multipartite systems, starting from a three-qubit system and beyond there exists at least one entangled state with the property that prevents any completely product state from yielding a product state when nontrivially superposed with the entangled state. However, in the bipartite Hilbert spaces $\mathbb{C}^2\otimes\mathbb{C}^d$, any entangled state cannot have a Schmidt rank higher than two. Consequently, for any entangled state from these bipartite Hilbert spaces, it is always possible to find a product state such that a nontrivial superposition of the two results in a product state. Thus, we argue that every multipartite Hilbert space contains at least one state that is unconditionally superposition-robust, a feature not shared by bipartite systems. This constitutes a sharp contrast between bipartite and multipartite quantum systems. (iv) There are several potential applications of the concept of the unconditional inseparability of superpositions. One such application is connected with a class of local state discrimination problems. 

We are now ready to present our findings. In Sec.~\ref{results_n_discussions}, we present our main results and corresponding discussions. In the same section, we also include the applications of our theory. We draw the conclusions of our work in Sec.~\ref{sec:Conclusion}. 

\section{Results and Discussions}\label{results_n_discussions}
Before we discuss our results, we note here that the following notations we use interchangeably throughout this paper, $\ket{A}\otimes\ket{B}\otimes\ket{C}\otimes\cdots\equiv\ket{A}\ket{B}\ket{C}\cdots\equiv\ket{ABC\dots}$ for multipartite completely product states. For all the results, when we say superposition, it is basically the nontrivial superposition, i.e., the coefficients corresponding to the states to be superposed, are nonzero (for example, the coefficients $a$, $b$ are nonzero in Fig.~\ref{fig:quantum_superposition}). Now, before revisiting a previously existing result, let us go through the types of states possible in multipartite systems:
\begin{definition}
m-separability~\cite{RPMK2009, GUHNE2009, ASU2023,STMAAU17}: An n-partite $(n\geq3)$ pure quantum state is termed as m-separable, if it can be written as a product of pure states of a maximum of m sub-systems, viz., $\ket{\Psi_{m-sep}}=\ket{\Psi_{1}}\otimes\ket{\Psi_{2}}\otimes....\ket{\Psi_{m}}$, where m$\le$n. The state is fully separable, i.e., completely product when $m=n$, and it is considered genuinely entangled when $m=1$ i.e., not bi-separable. 
\end{definition}
In \cite{sep+ent}, it was proved that pure bipartite states having Schmidt rank three or higher always outputs an entangled state when superposed with a product state, no matter what is the form of the product state. In this context, we mention that Schmidt rank of a bipartite entangled state is basically the minimum number of product states required to express the entangled state. However, one can do a straightforward generalization of the bipartite result of \cite{sep+ent} into a multipartite scenario. This can be stated as the following. 
\begin{proposition}\label{prop_a}
Given a multipartite Hilbert space, if the total dimension of the Hilbert space $D>8$, then it is always possible to find a multipartite entangled state which outputs an entangled state when superposed with any multipartite completely product state.\end{proposition}
\begin{proof}
If $D>8$, then we can always find a pure multipartite entangled state which have Schmidt rank three or higher in at least one bipartition. Thus, it always outputs an entangled state after getting superposed with an arbitrary completely product state \cite{sep+ent}. Here are few examples, elaborating above statement. (i) We consider $\ket{000}+\ket{101}+\ket{210}\in\mathbb{C}^3\otimes\mathbb{C}^2\otimes\mathbb{C}^2$. This state has Schmidt rank three in qutrit versus two-qubit bipartition. (ii) We consider $\ket{000}+\ket{111}+\ket{222}\in\mathbb{C}^3\otimes\mathbb{C}^3\otimes\mathbb{C}^3$. This state has Schmidt rank three across every bipartition. (iii) We consider
$\ket{0000}+\ket{0101}+\ket{1010}\in\mathbb{C}^2\otimes\mathbb{C}^2\otimes\mathbb{C}^2\otimes\mathbb{C}^2$. This state has Schmidt rank three in two-qubit versus two-qubit bipartition. (iv) In general, we can consider an n-partite Hilbert space $\mathbb{C}^{d_{1}}\otimes\mathbb{C}^{d_{2}}\otimes\mathbb{C}^{d_{3}}\otimes\cdots\otimes\mathbb{C}^{d_{n}}$, where $d_{1}\le d_{2}\le\cdots\le d_{n}$, $n\ge3$ and $d_1d_2\dots d_n>8$. Then, we can find a state $\ket{\Psi}$ in the given Hilbert space such that the state is entangled in $\mathbb{C}^{d_n}$ versus rest bipartition when $n=3$ or the state is entangled in $\mathbb{C}^{d_n}\otimes\mathbb{C}^{d_{n-1}}$ versus rest bipartition when $n>3$. These allow us to consider their Schmidt ranks in the bipartitions $\ge3$. These states if superposed with an arbitrary completely product state, always output an entangled state.
\end{proof}

Notice that the above observation about unconditional inseparability of superposition is clearly excluding the case where all the subsystems have dimension two. In this work, we particularly provide nontrivial results for this excluded case to fill the gap and complete the multipartite picture. Note that the term ``product state" specifically refers to a completely product state from now onward. For bi-separable states which are entangled in some but not in every bipartition, their use in this paper is explicitly stated whenever such states are referred. Now, if a pure multipartite state is entangled, it guarantees the existence of entanglement within the state in at least one bipartition. Suppose, we consider this specific bipartition. Then, the entanglement across this bipartition allows us to express the pure multipartite state in terms of its Schmidt decomposition~\cite{Schmidt1907,AK1995,A1995} in that bipartition. For the present analysis, we confine our study to three-qubit systems first. So, in the aforesaid bipartition, say A:BC, a three-qubit entangled state takes the following form:
\begin{equation}\label{eq:eqn1}
\ket{\Psi}=c\ket{\alpha}_{A}\ket{\Phi}_{BC}+ d\ket{\alpha^{\perp}}_{A}\ket{\Phi^{\perp}} _{BC},
\end{equation}
where $c,d$ are real, nonzero numbers and $c^2+d^2=1$. Let us consider that either $\ket{\Phi}$ or $\ket{\Phi^{\perp}}$ is entangled while the other is a product state. If both states are product states then, the superposition could yield a product state. However, our interest lies in the {\it unconditional inseparability} of superpositions. Thus, the aforementioned form of a multipartite pure entangled state is quite justified as a starting point for our analysis. It is essential to highlight that, given the existence of such Schmidt decomposition for any three-qubit pure entangled state, it is not required to consider the scenario where $\ket{\Phi}$ and $\ket{\Phi^{\perp}}$ are non orthogonal to each other.

\subsection{Separability of superposition}\label{subsec:Separability of superposition}
We consider a three-qubit pure entangled state $\ket{\Psi}$ and restate the following assumptions. 

\begin{itemize}
\item The state $\ket{\Psi}$ has the following form $c\ket{\alpha}\ket{\Phi} + d\ket{\alpha^\perp}\ket{\Phi^\perp}$, where $c,d$ are real, nonzero numbers, $c^2+d^2$ = 1, $\langle\alpha|\alpha^\perp\rangle$ = 0 = $\langle\Phi|\Phi^\perp\rangle$, in at least one bipartition.

\item This is basically Schmidt decomposition in this bipartition. Without loosing generality, we say this bipartition as A vs BC bipartition or A:BC bipartition.

\item We also consider that $\ket{\Phi}$ is an entangled state and $\ket{\Phi^\perp}$ is a product state. 
\end{itemize}

\noindent
Then, for such a state $\ket{\Psi}$, we have the following proposition.

\begin{proposition}\label{prop:prop1}
There exists a product state $\ket{\eta}\ket{\beta}\ket{\gamma}$, such that when superposed with $\ket{\Psi}$, the state after superposition might be a product state. In other words, the state $a\ket{\Psi}+b\ket{\eta}\ket{\beta}\ket{\gamma}$ might be a state of the form $\ket{\eta^\prime}\ket{\beta^\prime}\ket{\gamma^\prime}$, where $|a|^2+|b|^2 = 1$, $|a|, |b| \neq 0$.
\end{proposition}

\begin{proof}
To prove the above, it is sufficient to find an example. To construct such a nontrivial example, we start with the form of the state $\ket{\Psi}$ in A:BC bipartition and we take the state $\ket{\eta}\ket{\beta}\ket{\gamma}$ as $\ket{\alpha}\ket{\beta}\ket{\gamma}$. Then, we consider the superposition:

\begin{equation*}
\begin{aligned}
a\ket{\Psi} + b\ket{\alpha}\ket{\beta}\ket{\gamma} &= a(c\ket{\alpha}\ket{\Phi} + d\ket{\alpha^\perp}\ket{\Phi^\perp}) + b\ket{\alpha}\ket{\beta}\ket{\gamma}\\
                 &= \ket{\alpha}(ac\ket{\Phi}+b\ket{\beta}\ket{\gamma})+ad\ket{\alpha^\perp}\ket{\Phi^\perp}.
\end{aligned}
\end{equation*}
We assume that in A:BC bipartition the state is a product state i.e., $ac\ket{\Phi}+b\ket{\beta}\ket{\gamma} = ad\ket{\Phi^\perp}$. If this happens then, after superposition the state is proportional to $\ket{\eta^\prime}\ket{\Phi^\perp} \equiv \ket{\eta^\prime}\ket{\beta^\prime}\ket{\gamma^\prime}$, which is a completely product state. This implies that $\ket{\beta}\ket{\gamma} = a(d\ket{\Phi^\perp}-c\ket{\Phi})/b$. We can take $a = -b$. In this way, the problem boils down to the following problem. We have to find an example where superposition of a product state and an entangled state gives a product state. The entangled state and the product state must be orthogonal to each other. Thus, it is sufficient to show that $c\ket{\Phi}-d\ket{\Phi^\perp}$ is a product state.

Now, we consider $\ket{\Phi} = \sqrt{1/3}\ket{00}+\sqrt{2/3}\ket{11}$ and $\ket{\Phi^\perp} = (\sqrt{2/3}\ket{0}+\sqrt{1/3}\ket{1})(\sqrt{1/2}\ket{0}-\sqrt{1/2}\ket{1})$. It is easy to check that both states are orthogonal to each other. Next, we can consider that $c=1/\sqrt{5}$ and $d=2/\sqrt{5}$, for which the state $c\ket{\Phi}-d\ket{\Phi^\perp}$ is a product state. These complete the proof.
\end{proof}

Clearly, the assumptions we consider for $\ket{\Psi}$ before Proposition \ref{prop:prop1}, are not sufficient to declare that $a\ket{\Psi}+b\ket{\eta}\ket{\beta}\ket{\gamma}$ is always entangled, no matter what the form of the state $\ket{\eta}\ket{\beta}\ket{\gamma}$, where $|a|^2+|b|^2 = 1$, $|a|, |b| \neq 0$. So, here these assumptions are not useful to observe unconditional inseparability of superposition.

\subsection{Inseparability of superposition}\label{subsec:Inseparability of superposition}
We now introduce an additional assumption to the state $\ket{\Psi}$ of Eq.~(\ref{eq:eqn1}). This is given as the following. 

\begin{itemize}
\item Given that $\ket{\Phi}$ is a two-qubit entangled state, it has a decomposition as sum of two product states. Now, let us assume that $\ket{\Phi^{\perp}}$ is such that it is orthogonal to both of the product states in the decomposition of $\ket{\Phi}$.
\end{itemize}
Then, we have the following proposition.

\begin{proposition}\label{prop:prop2}
Under all of the assumptions corresponding to the state $\ket{\Psi}$, the three-qubit entangled state when superposed with any completely product state, always outputs an entangled state. 
\end{proposition}

\begin{remark}\label{remark:remark1}
Here the assumptions mean the assumptions which are described in the itemized forms just before the statements of Proposition \ref{prop:prop1} and Proposition \ref{prop:prop2}. Furthermore, `the output state is entangled' means it is entangled in at least one bipartition, i.e., it is not a completely product state.
\end{remark}

\begin{proof}
Consider a three qubit state $\ket{\Psi}$ having the form given in Eq.~(\ref{eq:eqn1}), i.e.,
\begin{equation*}
\ket{\Psi}=c\ket{\alpha}\ket{\Phi}+ d\ket{\alpha^{\perp}}\ket{\Phi^{\perp}}, 
\end{equation*}
where $c,d$ are real, nonzero numbers and $c^2+d^2=1$. We can write the entangled state $\ket{\Phi}$ in the Schmidt form as follows, $\ket{\Phi} = s\ket{00} + t\ket{11}$, where $s,t$ are real, nonzero numbers and $s^2+t^2=1$. Given this Schmidt decomposition, an appropriate choice for the product state $\ket{\Phi^{\perp}}$, which is orthogonal to the product states $\ket{00}$ and $\ket{11}$ present in the Schmidt decomposition of the entangled state $\ket{\Phi}$, would be $\ket{01}$ or $\ket{10}$.  

Now, we consider a completely product state $\ket{p}$ and present the proof in three different parts. For these parts, we consider three different forms of $\ket{p}$. They are: (i) $\ket{\alpha}\ket{\beta}\ket{\gamma}$, (ii) $\ket{\alpha^\perp}\ket{\beta}\ket{\gamma}$, (iii) $(x\ket{\alpha}+y\ket{\alpha^\perp})\ket{\beta}\ket{\gamma}$, when `$|x|$' and `$|y|$' both are nonzero. 

Suppose, we first assume $\ket{p}=\ket{\alpha}\ket{\beta}\ket{\gamma}$. Consider a nontrivial superposition of $\ket{\Psi}$ with $\ket{p}$, meaning that both `$|a|$' and `$|b|$' are nonzero and $|a|^2+|b|^2=1$ in the following.  
\begin{equation*}
\begin{aligned}
a\ket{\Psi} + b\ket{p} &= a(c\ket{\alpha}\ket{\Phi} + d\ket{\alpha^\perp}\ket{\Phi^\perp}) + b\ket{\alpha}\ket{\beta\gamma} \\
                       &= \ket{\alpha}(ac\ket{\Phi} + b\ket{\beta\gamma}) + ad\ket{\alpha^\perp}\ket{\Phi^\perp}.
\end{aligned}
\end{equation*}
Recall that $\ket{\Phi}$ is an entangled state and $\ket{\beta\gamma}$ is a product state. Therefore, their linear combination can give a product state but this product state can never be proportional to $\ket{\Phi^\perp}$. This can be easily seen from the fact that any linear combination of $\ket{\Phi}$ and $\ket{\Phi^\perp}$ always gives an entangled state, under our additional constraint. So, the output state of a nontrivial superposition of the state $\ket{\Psi}$ with a product state $\ket{p}$ can never be a completely product state.

Then, we consider $\ket{p} = \ket{\alpha^\perp}\ket{\beta}\ket{\gamma}$. So, now we have the following.
\begin{equation*}
\begin{aligned}
a\ket{\Psi} + b\ket{p} &= a(c\ket{\alpha}\ket{\Phi} + d\ket{\alpha^\perp}\ket{\Phi^\perp}) + b\ket{\alpha^\perp}\ket{\beta\gamma}\\
                       &= ac\ket{\alpha}\ket{\Phi} + \ket{\alpha^\perp}(ad\ket{\Phi^\perp} + b\ket{\beta\gamma}). 
\end{aligned}
\end{equation*}
To make the output a completely product state, one has to make $(ad\ket{\Phi^\perp} + b\ket{\beta\gamma})$ proportional to $\ket{\Phi}$. But this makes the output state a bi-separable state, not a completely product state.

Finally, we take $\ket{p} = (x\ket{\alpha}+y\ket{\alpha^\perp})\ket{\beta}\ket{\gamma}$, when `$|x|$' and `$|y|$' both are nonzero. So, after superposition, we have the following.
\small
\begin{equation*}
\begin{aligned}
a\ket{\Psi} + b\ket{p} &= 
a(c\ket{\alpha}\ket{\Phi} + d\ket{\alpha^\perp}\ket{\Phi^\perp}) + b(x\ket{\alpha}+y\ket{\alpha^\perp})\ket{\beta\gamma} \\
&=
\ket{\alpha}(ac\ket{\Phi} + bx\ket{\beta\gamma}) + \ket{\alpha^\perp}(ad\ket{\Phi^\perp} + by\ket{\beta\gamma}). 
\end{aligned}
\end{equation*}
\normalsize
Clearly, to make the output a completely product state, we have to make $ac\ket{\Phi} + bx\ket{\beta\gamma}$ proportion to $ad\ket{\Phi^\perp} + by\ket{\beta\gamma}$. So, basically the product state $\ket{\beta\gamma}$ can be written as some linear combination of $\ket{\Phi}$ and $\ket{\Phi^\perp}$. However, this combination cannot produce a product state as described previously. In this way, the proof is completed.
\end{proof}

\begin{remark}
We mention that through Proposition \ref{prop:prop2}, we achieve unconditional inseparability of superposition for the state $\ket{\Psi}$.
\end{remark}

Using the concept of tensor rank \cite{LERZA2010} of an entangled state, one may also think about doing the present analysis. Tensor rank is defined as the minimum number of completely product states required to express a multipartite entangled state. But tensor rank of an arbitrary state is hard to compute \cite{HASTAD1990644, Hiller13}. However, we bypass this difficulty and examine the underlying fact of unconditional inseparability of superposition. For this purpose, one may think about analyzing the form of $\ket{\Psi}$. This is exactly what we have done through Proposition \ref{prop:prop1} and Proposition \ref{prop:prop2}. Clearly, the way which we have followed here, may not be the only way for the present analysis but we argue that our way of the analysis might be easier and effective compared to another approach.

The key transition from Proposition \ref{prop:prop1} to Proposition \ref{prop:prop2} is based on the additional assumption regarding $\ket{\Phi}$ and $\ket{\Phi^\perp}$. This ensures the persistent presence of traces of entanglement in the output state. Interestingly, on the basis of this, it is also possible to differentiate between the states belonging to the GHZ (Greenberger-Horne-Zeilinger) class and those belonging to the W class (for these classes see Ref.~\cite{DVC00}). This differentiation can be treated as an application of our theory. Furthermore, this also justifies our approach of starting with the forms of the entangled states for the present analysis.

\subsection{Remaining part}
A three-qubit pure entangled state which is not genuinely entangled can always be written as linear combination of two completely product states. Therefore, it is always possible to find a completely product state, such that superposition of this product state and the entangled state which is not genuinely entangled, can produce completely product state. Then, we are left with pure genuinely entangled states which are entangled across every bipartition. They are either GHZ-type or W-type for three qubits \cite{DVC00}. Their roles in superposition are already discussed above. Thus, the discussion of three-qubit entangled states in superposition with completely product state is complete now. However, we have started with the form of the entangled states. Thus, for completeness, we discuss the following case. Let us now go back to the state given by Eq.~(\ref{eq:eqn1}) again, but now we examine the case where both $\ket{\Phi}$ and $\ket{\Phi^{\perp}}$ are orthogonal entangled states. We now present the following proposition.

\begin{proposition}\label{prop:prop3} 
We consider a state $\ket{\eta}\ket{\beta}\ket{\gamma}$, such that when superposed with $\ket{\Psi}$, in which both $\ket{\Phi}$ and $\ket{\Phi^{\perp}}$ are orthogonal entangled states, the state after superposition can either be a bi-separable state or a completely product state. 
\end{proposition}

\begin{proof}
Consider a three qubit state $\ket{\Psi}$ having the form given in Eq.~(\ref{eq:eqn1}), i.e.,
\begin{equation*}
\ket{\Psi}=c\ket{\alpha}\ket{\Phi}+ d\ket{\alpha^{\perp}}\ket{\Phi^{\perp}},
\end{equation*}
where both $\ket{\Phi}$ and $\ket{\Phi^{\perp}}$ are orthogonal entangled states, $c,d$ are real, nonzero numbers, $c^2+d^2=1$. Similar to Proposition \ref{prop:prop2}, we would like to analyze the three possible cases for any completely product state $\ket{p}$. This leads to the following observations. 

(i) We first assume $\ket{p}=\ket{\alpha}\ket{\beta}\ket{\gamma}$, then,
\begin{equation*}
\begin{aligned}
a\ket{\Psi} + b\ket{p} &= a(c\ket{\alpha}\ket{\Phi} + d\ket{\alpha^\perp}\ket{\Phi^\perp}) + b\ket{\alpha\beta\gamma}\\
&= \ket{\alpha}(ac\ket{\Phi} + b\ket{\beta\gamma}) + ad\ket{\alpha^\perp}\ket{\Phi^\perp}.
\end{aligned}
\end{equation*}
For such a resultant state to be product in A:BC bipartition, $ac\ket{\Phi} + b\ket{\beta\gamma}$ has to be proportional to $\ket{\Phi^\perp}$ but now given that $\ket{\Phi^\perp}$ is already a bipartite pure entangled state implies that the resultant state would always be an entangled state (entangled in bipartition(s) other than A:BC bipartition). So, here the state after superposition cannot be a completely product state.

(ii) Next, we assume $\ket{p}=\ket{\alpha^\perp}\ket{\beta}\ket{\gamma}$, then,
\begin{equation*}
\begin{aligned}
a\ket{\Psi} + b\ket{p} &= a(c\ket{\alpha}\ket{\Phi} + d\ket{\alpha^\perp}\ket{\Phi^\perp}) + b\ket{\alpha^\perp\beta\gamma} \\
&= \ket{\alpha^\perp}(ad\ket{\Phi^\perp} + b\ket{\beta\gamma}) + ac\ket{\alpha}\ket{\Phi}.
\end{aligned}
\end{equation*}
Following the same argument as in case of (i), again, we have an entangled state after superposition.

(iii) Finally, we assume $\ket{p}=(x\ket{\alpha}+y\ket{\alpha^\perp})\ket{\beta}\ket{\gamma}$ $|x|,|y|>0$, and $|x|^2+|y|^2=1$ then,
\begin{equation*}
\begin{aligned}
& a\ket{\Psi} + b\ket{p} & \\
& = a(c\ket{\alpha}\ket{\Phi} + d\ket{\alpha^\perp}\ket{\Phi^\perp}) + b(x\ket{\alpha}+y\ket{\alpha^\perp})\ket{\beta\gamma} & \\
& = \ket{\alpha^\perp}(ad\ket{\Phi^\perp} + by\ket{\beta\gamma}) + \ket{\alpha}(ac\ket{\Phi}+bx\ket{\beta\gamma}). &
\end{aligned}
\end{equation*}
If the resultant state is product in A:BC bipartition, then $ad\ket{\Phi^\perp} + by\ket{\beta\gamma}$ is proportional to $ac\ket{\Phi}+bx\ket{\beta\gamma}$ which implies that $\ket{\beta\gamma}$ is some linear combination of $\ket{\Phi}$ and $\ket{\Phi^\perp}$. So, this case can give a completely product output. 
\end{proof}
In the above, case (iii) can be understood by considering examples.  
\begin{example}Consider $\ket{\Phi} = \frac{1}{\sqrt{2}}(\ket{00}+\ket{11})$ and $\ket{\Phi^\perp} = \frac{1}{\sqrt{2}}(\ket{00}-\ket{11})$. Then, we take $\ket{\Psi} = \frac{1}{\sqrt{2}}(\ket{\alpha}\ket{\Phi}+\ket{\alpha^\perp}\ket{\Phi^\perp})$. Next, we take $\ket{p} = \frac{1}{\sqrt{2}}(\ket{\alpha}+\ket{\alpha^\perp})\ket{00}$. Finally, in superposition, we consider the coefficients as $a = \sqrt{2}$ and $b = -1$. In this way, we can obtain a state after superposition as $\frac{1}{\sqrt{2}}(\ket{\alpha}-\ket{\alpha^\perp})\ket{11}$.
\end{example}

Note that these observations remain the same even if we choose $\ket{\Phi}$ and $\ket{\Phi^\perp}$ to be bi-orthogonal, i.e., the product states in the decomposition of $\ket{\Phi}$ are orthogonal to that of $\ket{\Phi^\perp}$. In that case also one might get output to be completely product. 
\begin{example}
Consider $\ket{\Phi} = \frac{1}{\sqrt{2}}(\ket{00}+\ket{11})$ and $\ket{\Phi^\perp} = \frac{1}{\sqrt{2}}(\ket{01}+\ket{10})$, $\ket{\Psi} = \frac{1}{\sqrt{2}}(\ket{\alpha}\ket{\Phi}+\ket{\alpha^\perp}\ket{\Phi^\perp})$. Then, the product state is $\ket{p} = \frac{1}{\sqrt{2}}(\ket{\alpha}-\ket{\alpha^\perp})\ket{-}\ket{-}$, where $\ket{-} = \frac{1}{\sqrt{2}}(\ket{0}-\ket{1})$. Finally, we take $a = \sqrt{2}$ and $b = -1$. Applying all these, one can get a completely product state after superposition.
\end{example}

\subsection{Bipartite systems versus multipartite systems}
A multipartite system can be different from a bipartite system in various ways. Here we describe one such instance. This is of course based on the findings which is described so far. We are now ready to present the following theorem.

\begin{theorem}\label{thm1}
There always exists at least one multipartite entangled state, irrespective of the structure of the composite quantum system, such that any nontrivial superposition of the state with an arbitrary completely product state always yields a multipartite entangled state. This is in sharp contrast with the bipartite systems.
\end{theorem}

\noindent
{\bf Note:} This theorem suggests that `unconditional inseparability of superposition' (as defined in this work) exists in all multipartite Hilbert spaces. But in case of bipartite systems this is not true in general. Let us now proceed towards the proof.

\begin{proof}
From Proposition \ref{prop:prop2}, it is evident that one can get unconditional inseparability of superposition in the minimum dimensional multipartite Hilbert space. If the dimension or the number of parties is increased , then, there are states on which one can apply the result of \cite{sep+ent}. For example in $\mathbb{C}^3\otimes\mathbb{C}^2\otimes\mathbb{C}^2$,  we can have states like $\ket{000}+\ket{101}+\ket{210}$ which has schmidt rank 3 in A:BC bipartion . Also in $\mathbb{C}^2\otimes\mathbb{C}^2\otimes\mathbb{C}^2\otimes\mathbb{C}^2$,  we can have states like $\ket{0000}+\ket{0101}+\ket{1010}$ which again have schmidt rank 3 in AB:CD bipartition. This is how one can obtain entanglement after superposing with an arbitrary completely product state. Such examples are also explained at the beginning of this section (Sec.~\ref{results_n_discussions}). In particular, see Proposition \ref{prop_a}. In this way, we get the proof for the first part of the theorem.

On the other hand, in the bipartite scenario when one of the subsystems is a qubit, the existence of the Schmidt decomposition guarantees the following. For any entangled state, one can always obtain a product state as an output by appropriately selecting the product state in the superposition. This is because in the Schmidt decomposition of a qubit-qudit pure entangled state there can be only two product states. Thus, unconditional inseparability of superposition is not possible when the quantum system corresponds to the Hilbert space, $\mathcal{H} = \mathbb{C}^2\otimes\mathbb{C}^d$. This completes the proof.   
\end{proof}

Three-qubit case is the only case where it is not possible to get any bipartition where we can accommodate states with Schmidt rank three or higher. Apart from this, any multipartite configuration we consider, there will be at least one bipartition where we can accommodate states with Schmidt rank three or higher. Then one can apply a result of \cite{sep+ent}, see Proposition \ref{prop_a}.

\subsection{Bi-separable states replacing the completely product states}\label{subsec:Inseparability of states when superposed with bi-separable state}
Here we consider the superposition of a pure multipartite entangled state (preferably a genuinely entangled state) and a bi-separable state. Then, we ask when the state after superposition is always entangled, i.e., not a completely product state. We mention that a pure state is genuinely entangled state if it is entangled across every bipartition. Furthermore, a pure bi-separable state is a state which is neither a completely product state nor a genuinely entangled state.

A straightforward exemplar elucidates how any state $\ket{\Psi}$, articulated by the form in Eq.~(\ref{eq:eqn1}), can yield a completely product state after undergoing a nontrivial superposition with a bi-separable state $\ket{p} = \ket{\alpha}\ket{\chi}$, product in the A:BC bipartition. This can be understood as the following. Consider a bi-separable state $\ket{p} = \ket{\alpha}\ket{\Phi}$, where $\ket{\Phi}$ is same entangled state as in Eq.~(\ref{eq:eqn1}). Upon superposing the state $\ket{\Psi}$ with such a bi-separable state, we obtain $a_{1}(c\ket{\alpha}\ket{\Phi}+ d\ket{\alpha^{\perp}}\ket{\Phi^{\perp}}) + a_2\ket{\alpha}\ket{\Phi}$. By choosing $a_1, a_2, c$ such that $a_1c = -a_{2}$, we achieve the output state as $\ket{\alpha^{\perp}}\ket{\Phi^{\perp}}$, which is a completely product state, provided $\ket{\Phi^\perp}$ is a product state. Notice that bi-orthogonality of $\ket{\Phi}$ and $\ket{\Phi^{\perp}}$ is not required for the existence of such examples. We mention that in this subsection, we consider bi-separable states that are product in A:BC bi-partition only.

Now, we modify the question a bit. Does there exist a bi-separable state $\ket{p}=\ket{\alpha}\ket{\chi}$ such that the entangled state $\ket{\chi}$, not proportional to $\ket{\Phi}$ of Eq.~(\ref{eq:eqn1}), still a nontrivial superposition of $\ket{\Psi}$ and $\ket{p}$ results in a completely product state?
 
\begin{proposition}\label{prop:prop4}
There may exist a bi-separable state $\ket{\eta}\ket{\chi}$, where $\ket{\chi}$ is an entangled state, such that when superposed with $\ket{\Psi}$ of the form given in Eq.~(\ref{eq:eqn1}), the state after superposition is a completely product state. In other words, the state $a\ket{\Psi}+b\ket{\eta}\ket{\chi}$ might be a state of the form $\ket{\eta^\prime}\ket{\beta^\prime}\ket{\gamma^\prime}$, where $|a|^2+|b|^2 = 1$, $|a|, |b| \neq 0$.
\end{proposition}

\noindent
{\bf Note:} Here we assume that $\ket{\Phi}$ is an entangled state and $\ket{\Phi^\perp}$ is a product state in Eq.~(\ref{eq:eqn1}) and $\langle\Phi|\Phi^\perp\rangle = 0$. Again, $\ket{\chi}$ is not proportional to $\ket{\Phi}$.

\begin{proof}
Let us assume that the bi-separable state is of the form $\ket{p} = \ket{\alpha}_{A}\ket{\chi}_{BC}$. Then, we consider nontrivial superposition of $\ket{p}$ with the state of Eq.~(\ref{eq:eqn1}), i.e.,
\begin{equation*}
\begin{aligned}
a\ket{\Psi} + b\ket{p} &= a(c\ket{\alpha}\ket{\Phi} + d\ket{\alpha^\perp}\ket{\Phi^\perp}) + b\ket{\alpha}\ket{\chi}\\
                       &= \ket{\alpha}(ac\ket{\Phi} + b\ket{\chi}) + ad\ket{\alpha^\perp}\ket{\Phi^\perp}.
\end{aligned}
\end{equation*}
For the purpose of constructing an example, let $b=-a$, such that appropriate normalization factor can be inserted in the superposition. We assume that in A:BC bipartition the state $a\ket{\Psi} + b\ket{p}$ is a product state, this implies that $c\ket{\Phi} -\ket{\chi} = d\ket{\Phi^{\perp}}$. Therefore, from such an example, we can observe that states of the form as given in Eq.~(\ref{eq:eqn1}), when superposed with bi-separable state, can output completely product state. Hence, the proof is complete.
\end{proof}
Relevant examples in support of the above proposition are given as the following.

\begin{example}
Consider $\ket{\Phi}=\frac{1}{\sqrt{2}}(\ket{00}-\ket{11})$ and $\ket{\chi}=(\frac{c}{\sqrt{2}}+\frac{d}{2})\ket{00}+\frac{d}{2}\ket{01}+\frac{d}{2}\ket{10}+(\frac{-c}{\sqrt{2}}+\frac{d}{2})\ket{11}$. It can be seen easily that, for such $\ket{\Phi}$ and $\ket{\chi}$, $c\ket{\Phi} -\ket{\chi}$ turns out to be proportional to $\ket{++}$, orthogonal to $\ket{\Phi}$. 
\end{example}

Let us take another example related to above proposition.
\begin{example}
$\ket{\Phi} = (1/\sqrt{2})(\ket{00} + \ket{11})$ and $\ket{\chi} = (1/\sqrt{3})(\ket{00} + \ket{11} + \ket{01})$. Then, for $(c\ket{\Phi} - \ket{\chi})$, we choose $c = \sqrt{2}/\sqrt{3}$, and we get, $(c\ket{\Phi} - \ket{\chi})$ = $\ket{01}$. This state is clearly orthogonal (and also bi-orthogonal) to $\ket{\Phi}$.
\end{example}
However, if we make the following assumption that $\ket{\Phi}$ and $\ket{\chi}$ to be bi-orthogonal entangled states, in addition to previously mentioned assumptions then, we have the following proposition.

\begin{proposition}\label{prop:prop5}
Under all of the assumptions corresponding to the state $\ket{\Psi}$, the three-qubit entangled state when superposed with a bi-separable state, always outputs an entangled state.
\end{proposition}

\noindent
{\bf Note:} In brief we mention the assumptions corresponding to $\ket{\Psi}$. The form of $\ket{\Psi}$ is given in Eq.~(\ref{eq:eqn1}). The state $\ket{\Phi}$ is an entangled state and $\ket{\Phi^\perp}$ is a product state, $\langle\Phi|\Phi^\perp\rangle = 0$. Furthermore, $\ket{\Phi^\perp}$ is orthogonal to the product states that appear in the Schmidt decomposition of $\ket{\Phi}$. There is also restriction on the form of $\ket{\chi}$. The product states that appear in the Schmidt decomposition of $\ket{\chi}$ are orthogonal to that of $\ket{\Phi}$.

\begin{proof}
We consider a three qubit state $\ket{\Psi}$ having the form given in Eq.~(\ref{eq:eqn1}), i.e., $\ket{\Psi}=c\ket{\alpha}\ket{\Phi}+ d\ket{\alpha^{\perp}}\ket{\Phi^{\perp}}$, where $c,d$ are real, nonzero numbers and $c^2+d^2=1$. We can write the entangled state $\ket{\Phi}$ in the Schmidt form as $\ket{\Phi} = s\ket{00} + t\ket{11}$, where $s,t$ are real, nonzero numbers and $s^2+t^2=1$. So, the product state $\ket{\Phi^{\perp}}$ can be either $\ket{01}$ or $\ket{10}$. Next, we consider a bi-separable state $\ket{p}$. We consider the form of $\ket{p}$ as $(x\ket{\alpha}+y\ket{\alpha^\perp})\ket{\chi}$, $|x|\geq0$, $|y|\geq0$, and $|x|^2+|y|^2=1$. $\ket{\chi}$ is an entangled state. Because of the bi-orthogonality of $\ket{\Phi}$ and $\ket{\chi}$, we can assume $\ket{\chi}$ to be of the form $p\ket{01}+q\ket{10}$, where $p,q$ are real, nonzero numbers and $p^2+q^2=1$. 

Now, we consider the following:
\begin{equation*}
\begin{aligned}
& a\ket{\Psi} + b\ket{p} & \\
&= a(c\ket{\alpha}\ket{\Phi} + d\ket{\alpha^\perp}\ket{\Phi^\perp}) + b(x\ket{\alpha}+y\ket{\alpha^\perp})\ket{\chi} &\\
&= \ket{\alpha}(ac\ket{\Phi} + bx\ket{\chi}) + \ket{\alpha^\perp}(ad\ket{\Phi^\perp} + by\ket{\chi}), &
\end{aligned}
\end{equation*}
where `$|a|$' and `$|b|$' are nonzero and $|a|^2+|b|^2=1$. According to our assumptions, it is easy to check that $(ac\ket{\Phi} + bx\ket{\chi})$ and $(ad\ket{\Phi^\perp} + by\ket{\chi})$ are linearly independent. Therefore, the output state is entangled in at least
one bipartition. This completes the proof.
\end{proof}

It is essential to mention here the following important point. In the superposition $a\ket{\Psi} + b\ket{p}$, when we consider $\ket{p}$ as a completely product state, we have not put any assumption on the form of $\ket{p}$. In this way, we have achieved unconditional inseparability of superposition for certain class of multipartite entangled states $\ket{\Psi}$. On the other hand, when $\ket{p}$ is a bi-separable state, we have put additional assumption on the form of $\ket{p}$. We say this as weaker scenario to achieve entanglement after superposition. We are now ready to discuss some applications of our theory.

\subsection{Applications}\label{sec:Multipartite UEB and unambiguous local discrimination of quantum states}
We first recall the definition of unextendible entangled basis (UEB) from Refs.~\cite{HALDER2021168550, SJ2011, SU2022}. 
\begin{definition}
A UEB is defined as a collection of mutually orthonormal pure entangled states within a composite Hilbert space, such that the orthogonal complement of the subspace, spanned by the considered entangled states, contains no entangled states.
\end{definition}
In other words, the UEB constitutes a set of pure entangled states for which any additional pure state orthogonal to the states of this set must be a product state. We now focus on unextendible entangled bases (UEBs) that consist of genuinely entangled pure states. Furthermore, the complementary subspace here includes only completely product state. 

Let us give a construction of a special UEB in $\mathbb{C}^{3}\otimes\mathbb{C}^{3}\otimes\mathbb{C}^{3}$ which we say a 3-UEB. This construction is inspired from a bipartite construction, given in \cite{sep+ent}. We first identify eighteen entangled states $\{\ket{\psi_1},\dots,\ket{\psi_{18}}\}$ which have Schmidt rank three in at least one bipartition. These states are the states of the set, \{$\ket{\psi_1}$, $\ket{\psi_2}$, $\ket{\psi_3}$\}, belonging to the subspace spanned by the states in \{$\ket{000}$, $\ket{011}$, $\ket{122}$\}, the states in \{$\ket{\psi_4}$, $\ket{\psi_5}$, $\ket{\psi_6}$\}, belonging to the subspace spanned by \{$\ket{100}$, $\ket{111}$, $\ket{022}$\}, the states in \{$\ket{\psi_7}$, $\ket{\psi_8}$, $\ket{\psi_9}$\}, belonging to the subspace spanned by \{$\ket{001}$, $\ket{012}$, $\ket{120}$\}, the states in \{$\ket{\psi_{10}}$, $\ket{\psi_{11}}$, $\ket{\psi_{12}}$\}, belonging to the subspace spanned by \{$\ket{101}$, $\ket{112}$, $\ket{020}$\}, the states in \{$\ket{\psi_{13}}$, $\ket{\psi_{14}}$, $\ket{\psi_{15}}$\}, belonging to the subspace spanned by \{$\ket{002}$, $\ket{010}$, $\ket{121}$\}, and the states in \{$\ket{\psi_{16}}$, $\ket{\psi_{17}}$, $\ket{\psi_{18}}$\}, belonging to the subspace spanned by \{$\ket{102}$, $\ket{110}$, $\ket{021}$\}. In fact, these eighteen states can be orthogonal to each other. We can give an example for illustration -- $\ket{\psi_1}, \ket{\psi_2}, \ket{\psi_3}$ can be of the forms $(1/\sqrt{3})(\ket{000}+\ket{011}+\ket{122})$, $(1/\sqrt{3})(\ket{000}+\omega\ket{011}+\omega^2\ket{122})$, $(1/\sqrt{3})(\ket{000}+\omega^2\ket{011}+\omega\ket{122})$, where $\omega$ and $\omega^2$ are cubic roots of unity. The other states $\{\ket{\psi_4},\dots,\ket{\psi_{18}}\}$ can also be constructed in the same way. Next, we consider the product states \(\ket{200}\), \(\ket{211}\), \(\ket{201}\), \(\ket{220}\), \(\ket{210}\), \(\ket{221}\), \(\ket{222}\), \(\ket{212}\), \(\ket{202}\). These product states are also orthogonal to the previously mentioned eighteen entangled states. We now present the final basis in $\mathbb{C}^{3}\otimes\mathbb{C}^{3}\otimes\mathbb{C}^{3}$. It consists of the states: \{$\frac{1}{\sqrt{2}}(\ket{\psi_1}\pm\ket{200})$,
$\frac{1}{\sqrt{2}}(\ket{\psi_2}\pm\ket{211})$,
$\frac{1}{\sqrt{2}}(\ket{\psi_3}\pm\ket{201})$,
$\frac{1}{\sqrt{2}}(\ket{\psi_4}\pm\ket{220})$,
$\frac{1}{\sqrt{2}}(\ket{\psi_5}\pm\ket{210})$,
$\frac{1}{\sqrt{2}}(\ket{\psi_6}\pm\ket{221})$,
$\ket{\psi_7}$, $\ket{\psi_8}$,$\ket{\psi_9}$, $\ket{\psi_{10}}$, $\ket{\psi_{11}}$, $\ket{\psi_{12}}$, $\ket{\psi_{13}}$, $\ket{\psi_{14}}$, $\ket{\psi_{15}}$, $\ket{\psi_{16}}$, $\ket{\psi_{17}}$, $\ket{\psi_{18}}$, $\ket{222}$, $\ket{212}$, $\ket{202}$\}. Clearly, the first twenty-four states form a 3-UEB in $\mathbb{C}^{3}\otimes\mathbb{C}^{3}\otimes\mathbb{C}^{3}$ with the property that the entangled states have Schmidt rank three in at least one bipartition. Next, we proceed to prove a type of indistinguishability exhibited by this 3-UEB under LOCC. For this we basically consider state discrimination problem under LOCC with the unambiguous setting. In the following, it is described in a greater details. We note here that for this indistinguishability property incomplete sets are particularly important. For a complete set, finding such indistinguishability is trivial.

The local indistinguishability property via the unambiguous setting of state discrimination problem under LOCC, in which we are interested, can be described as the following.
\begin{definition}{Local Indistinguishability:}
Suppose, a composite quantum system is prepared in a particular state. This state is secretly chosen from a given set of orthogonal pure states which does not span a whole Hilbert space. Now we want to find such an incomplete set for which the state of the quantum system cannot be locally identified unambiguously with nonzero probability irrespective of the fact that in which state of the set the system is prepared. 
\end{definition}

The 3-UEB that we have constructed here is such an incomplete set. The proof of this directly follows from \cite{sep+ent} (in particular, see Theorem 3 of \cite{sep+ent}). For completeness, we add a few steps. To identify the state of a quantum system unambiguously by LOCC, it is necessary and sufficient to find a completely product state which should be nonorthogonal to the state to be identified but orthogonal to the other states of the given set \cite{CHEFLES2004}. Now, we consider a state taken from the aforesaid 3-UEB. In this state, a quantum system is prepared. It is required to identify the state unambiguously by LOCC. Clearly, for this, it is required to find a completely product state in the space spanned by the chosen state and the product states $\{\ket{222},~ \ket{212},~ \ket{202}\}$, which is nonorthogonal to the chosen state. Now, any linear combination of the states $\{\ket{222},~ \ket{212},~ \ket{202}\}$ produces another completely product state which always produces entangled state if superposed with the chosen state. This is because the chosen state has Schmidt rank three in at least one bipartition \cite{sep+ent}. So, it is impossible to identify the state of the quantum system unambiguously by LOCC. In fact, this is true for all states of the 3-UEB.

However, the states within a UEB in $\mathbb{C}^{2}\otimes\mathbb{C}^{2}\otimes\mathbb{C}^{2}$ cannot have Schmidt rank three or higher across any bipartition. So, apparently, if we stick to the results of \cite{sep+ent}, then, it is not possible to find an incomplete set in $\mathbb{C}^{2}\otimes\mathbb{C}^{2}\otimes\mathbb{C}^{2}$ for which the aforesaid indistinguishability property can be obtained. Here comes the importance of present analysis. We want to show that the present analysis of unconditional inseparability of superposition is useful to obtain the aforesaid indistinguishability property for certain sets in $\mathbb{C}^{2}\otimes\mathbb{C}^{2}\otimes\mathbb{C}^{2}$. For example, we consider the following W-UEB consisting of W-type genuinely entangled states \cite{HALDER2021168550}:
\begin{equation}
\begin{array}{l}
\frac{1}{\sqrt{3}}(\ket{001}+\ket{010}+\ket{100}),\\[1.5 ex]
\frac{1}{\sqrt{3}}(\ket{001}+\omega\ket{010}+\omega^2\ket{100}),\\[1.5 ex]
\frac{1}{\sqrt{3}}(\ket{001}+\omega^2\ket{010}+\omega\ket{100}),\\[1.5 ex]
\frac{1}{\sqrt{3}}(\ket{000}+\ket{101}+\ket{110}),\\[1.5 ex]
\frac{1}{\sqrt{3}}(\ket{000}+\omega\ket{101}+\omega^2\ket{110}),\\[1.5 ex]
\frac{1}{\sqrt{3}}(\ket{000}+\omega^2\ket{101}+\omega\ket{110}),
\end{array}
\end{equation} 
where $\omega$ and $\omega^2$ are cubic root of unity. Observe that each state in the aforementioned example of the W-UEB satisfies all the assumptions for $\ket{\Psi}$ associated with Proposition \ref{prop:prop2}. The complementary subspace of the completely product states is spanned by the states \{$\ket{011}$, $\ket{111}$\}. Furthermore, there can be many completely product states in this two-dimensional subspace. So, it is difficult to define the explicit form of an arbitrary state picked from this subspace. This W-UEB is an example of an incomplete set which possess the aforesaid indistinguishability property. More generally, we prove the following theorem.

\begin{theorem}\label{thm2}
We consider any UEB in $\mathbb{C}^{2}\otimes\mathbb{C}^{2}\otimes\mathbb{C}^{2}$ that includes entangled states, where each state satisfies all previously specified assumptions for $\ket{\Psi}$, associated with Proposition \ref{prop:prop2}. Then, such a UEB does not contain any state that is locally unambiguously identifiable.
\end{theorem}

\noindent
{\bf Note:} The above theorem depicts that if a quantum system is prepared in any state which is secretly chosen from such a UEB, then it is not possible to identify the state of the system unambiguously by LOCC with nonzero probability. We also mention that all pure states of the complementary subspace of such a UEB, are completely product. 

\begin{proof}
Let a state $\ket{\psi}$ is picked from such a UEB, i.e., $\ket{\psi}$ has the properties of $\ket{\Psi}$ as mentioned in Proposition \ref{prop:prop2} and \{$\ket{p_1}$, $ \ket{p_2}$, \ldots, $\ket{p_n}$\} be the completely product basis for the complementary subspace of the UEB. By the definition of present UEBs, any linear combination of these completely product states remains a completely product state. Therefore, for any state $\ket{p}$ in the span of the complementary subspace, the state $\ket{\phi} = a_0\ket{\psi} + a_1\ket{p}$, where $a_0$ and $a_1$ are nonzero complex numbers and $\langle \phi | \phi \rangle = 1$, always be an entangled state and never be a completely product state. This is because of Proposition \ref{prop:prop2}. Thus, no completely product state can be found that is nonorthogonal to $\ket{\psi}$ while being orthogonal to all other states of the considered UEB. But it is known to be necessary for unambiguous local identification with nonzero probability \cite{CHEFLES2004}. This is already mentioned during the discussion of 3-UEB. This completes the proof.
\end{proof}

\begin{corollary}\label{coro}
In general, be it three-qubit system or a system with higher dimension, if a UEB contains a state associated with Proposition \ref{prop:prop2} or a state which has Schmidt rank three (or higher) in at least one bipartition, then the state is not locally unambiguously identifiable with nonzero probability.
\end{corollary}
Corollary \ref{coro} basically provides a methodology to produce incomplete sets which are locally unambiguously indistinguishable and is sufficiently general.

We now discuss a few important points. These are related to the indistinguishability property mentioned in Theorem \ref{thm2}.

\begin{itemize}
\item To prove the indistinguishability property of the W-UEB, one can follow the proof technique of the above theorem (Theorem \ref{thm2}).

\item The Theorem \ref{thm2} is also important from the following aspect. In case of two qubits, there is no incomplete set which has the aforesaid indistinguishability property \cite{SU2022}. Clearly, the minimum number of qubits required to exhibit this indistinguishability property is three. Furthermore, finding a UEB in $\mathbb{C}^2\otimes\mathbb{C}^d$, $(d>2)$ with the indistinguishability property of Theorem \ref{thm2} is very difficult. This difficulty is due to Theorem \ref{thm1}.

\item Here we have talked about how our theory is becoming useful to exhibit a type of indistinguishability property under LOCC. In particular, Proposition \ref{prop:prop2} is providing a sufficient condition for exhibiting such a property. This is depicted in the above theorem. 

\item Here one may also expect some applications of the weaker superposition scenario (Proposition \ref{prop:prop5}). However, we are yet to find such an application. In this context, we like to mention the resource theory of \cite{PCJ19}. Using the free operation of this resource theory, one may think about exhibiting the present indistinguishability property. In that case, Proposition \ref{prop:prop5} might be useful.
\end{itemize}

\section{Conclusion}\label{sec:Conclusion}
There are several useful quantitative bounds available in the literature on the entanglement content of multiparty states obtained by superposing two or more entangled pure states. Now, whenever a nonzero lower bound is available for the output state of such a superposition (interference) process, the question of whether the same is entangled or not is obviously answered. Still, a qualitative approach can be useful in cases where quantitative approaches fail \cite{sep+ent}. Here, we have considered such a qualitative approach for the superposition of pure multipartite states and found that the results can be in stark contrast with those of the bipartite case. 

In particular, we found that there exists at least one pure entangled state in any multipartite Hilbert space such that its superposition with an arbitrary completely product state always yields an entangled state. Interestingly, for bipartite pure entangled states, if at least one of the subsystems is a qubit, then it is always possible to find a suitable product state such that the superposition of the entangled state and the product state is a product state. It is also interesting to note that our analysis can be used to differentiate between the two disjoint classes of stochastic LOCC-inequivalent genuine three-qubit entangled pure states, viz., the GHZ and the W classes. Finally, we have discussed how our analysis provides insights into the unambiguous discrimination of quantum states using local quantum operations and classical communication. 

\section*{Acknowledgments}
SH was funded by a postdoctoral fellowship of Harish-Chandra Research Institute, Prayagraj (Allahabad) when this project was started. 

SH is now funded by the European Union under Horizon Europe (grant agreement no.~101080086). Views and opinions expressed are however those of the author(s) only and do not necessarily reflect those of the European Union or the European Commission. Neither the European Union nor the granting authority can be held responsible for them.

We also acknowledge partial support from the Department of Science and Technology,
Government of India, through the QuEST grant with Grant
No. DST/ICPS/QUST/Theme-3/2019/120 via I-HUB QTF of IISER Pune, India.

\bibliography{ref}
\end{document}